%% file: arxiv.tex
\documentclass[runningheads]{llncs}

\usepackage{amsmath,amssymb}
\usepackage{booktabs}
\usepackage{graphicx}
\usepackage{threeparttable}
\usepackage{caption}
\usepackage{xcolor}
\usepackage[hidelinks]{hyperref}
\usepackage{subfiles}

\def\finalversion{}
\ifdefined\finalversion
    \newcommand{\edmark}[3]{}
\else
    \newcommand{\edmark}[3]{{\color{#2}\textbf{[#1: #3]}}}
\fi

\begin{document}

\title{Priority Gas Auction Cadence and Searcher Competition: Evidence from Flashblocks on Base}
\titlerunning{Priority Gas Auction Cadence and Searcher Competition}
\author{Pranay Anchuri, Edward W. Felten, Akaki Mamageishvili}
\authorrunning{Anchuri, Felten, Mamageishvili}
\institute{Offchain Labs}

\maketitle

\begin{abstract}
Flashblocks divide a block's priority gas auction into shorter sequential auctions that commit transaction order before the block is complete. We ask how this auction cadence affects bidding and competition among automated arbitrageurs, or searchers. In July 2025 Base replaced a single 2~s auction with ten 200~ms auctions.
We use this change to estimate the searcher response from on-chain data alone. On an address-day panel of $3{,}032$ searchers and $8{,}053$ activity-matched controls, a difference-in-differences design estimates a $0.187$~gwei fall in the effective priority fee ($59\%$ of the searcher pre-period mean). The share of searcher priority-fee value paid in the first tenth of block gas falls from $0.98$ to $0.28$. This fee compression is the identified effect. The revert-rate response is not (causally) identified and is reported descriptively. At the level of a single opportunity, the winning fee splits into an approximately invariant floor and a competitive premium. The premium scales about linearly with auction duration at high contention and less steeply at low contention. A first-price auction model with uncertain arrival and partial payments by losing searchers accounts for this pattern. We match attempts to the pool they contest. A higher-fee transaction reverts behind a lower-fee winner that landed earlier in $1.9$ to $2.7\%$ of contested opportunities by value, an upper bound. Latency therefore decides part of the ordering. A micro-auction model projects diminishing returns from shortening the window below 200~ms, as bids approach a nonzero floor. Rollups such as Arbitrum One now make auction cadence a governance tunable parameter, so these estimates inform how it is set.
\end{abstract}

\section{Introduction}
\label{sec:intro}
On Base, an Ethereum layer-2 rollup, the sequencer orders the transactions within a block in descending order of the priority fee per gas they bid, an ordering rule commonly referred to as {\it priority gas auction} (PGA). In July 2025 Base introduced {\it flashblocks}, streaming ten preconfirmation sub-blocks of 200~ms within each 2~s block. Each flashblock locks a transaction order as it is produced. The canonical block time stays at 2~s~\footnote{A 200~ms block time is planned to become canonical with an upcoming fork in October 2026: \url{https://docs.base.org/specifications/flashblocks}}. The change leaves the fee mechanism intact and alters only its frequency, converting a single per-block PGA into ten sequential 200~ms sub-auctions.

Automated arbitrageurs, or searchers, are the primary bidders in this auction and pay most of the priority fee. Decentralized-exchange transactions account for one fifth of transaction count on Base yet they pay $70\%$ of priority fees in June 2025 and $58\%$ in August 2025. The effect (if any) of restructuring the auction into faster sub-auctions on searcher bids, and on who wins, has not been measured. A single before-and-after comparison cannot answer it. The transition coincides with a market shift, so the pooled mean effective priority fee falls from $0.038$ to $0.025$~gwei for reasons that mix flashblocks with everything else.

Quantifying the searcher response is useful for three reasons. First, the priority fee that searchers pay is a cost of arbitrage that ordinary users ultimately bear through worse execution prices, so its level matters for them. Second, preconfirmation sub-blocks are being adopted across rollups, and their effect on the priority gas auction is a live design question rather than a settled one. Arbitrum One voted to replace Timeboost with a priority gas auction in September 2026, with a $125$~ms round length adjustable between $25$ and $250$~ms~\cite{arbitrum2026pga}. Auction cadence is therefore a governance tunable parameter. Third, it is not obvious whether faster sub-auctions raise or lower on-chain contention, which we examine descriptively.

We identify the effect of restructuring PGA with a difference-in-differences design on an address-day panel. For this analysis, active searchers form the treated group and a set of activity-matched but non-competitive addresses form the control group. August 2025 is the post period, compared against the pre-treatment period of June 2025. By using two-way fixed effects, the model absorbs every fixed difference between the groups and every shock common to both. Therefore, the identified coefficient recovers the searcher-specific response, provided the two groups would have moved together absent flashblocks. We test this parallel-trends assumption in Section~\ref{sec:ident}. Flashblocks leave no on-chain marker, so we reconstruct the 200~ms sub-block structure from the within-block priority-fee sequence. The reason for no on-chain marker is that flashblocks do not strictly order transactions by decreasing priority fee. Instead, the sequencer applies a heuristic~\footnote{Based on documentation (\url{https://docs.base.org/specifications/transactions/transaction-ordering}), the sequencer starts executing transactions in a sorted list of transactions in the flashblock, while it still accepts incoming transactions: if a newly arriving transaction pays a higher priority fee than the highest paying transaction in the remaining list, it is given priority, instead of being deferred to the next flashblock.}. We have no access to the sequencer, so every measurement in this paper is drawn from on-chain data alone.

\smallskip\noindent\textbf{Contributions.}
\begin{enumerate}
  \item \textbf{A flashblock reconstruction from on-chain data.} We develop a heuristic to recover the 200~ms sub-block
        boundaries from priority fee jumps within the block.  We validate the heuristic against the pre-flashblock month, where every block collapses to a single slice.
        \item \textbf{Causal estimates of the searcher response.} Under flashblocks the searcher effective priority fee falls by $0.187$~gwei ($0.168$ before panel weighting) relative to a matched control, $59\%$ of its pre-period level, and the declared bid by $0.228$~gwei ($0.183$ before weighting). The share of priority fees paid at the top of the block, an index defined in Section~\ref{sec:conc}, falls from $0.98$ to $0.28$.
        \item \textbf{Opportunity-level evidence on latency.} By matching competing transactions to the pool they contest, a higher-bidding transaction loses because of late arrivals in $1.9$ to $2.7\%$ of contested opportunities weighted by value, and $5.4$ to $6.6\%$ unweighted, both upper bounds. So, ordering within a block is decided partly by speed rather than price.
        \item \textbf{Regime stability and a shorter-window projection.} The block-level flashblock regime persists in October 2025 and June 2026, eleven months after launch. A small auction model fit to the two observed block times gives a rough bound below 200~ms, not a firm prediction, since no chain with PGA runs below 200~ms today~\footnote{Robinhood chain runs at 100~ms block times, but applies FCFS transaction ordering policy.}.
        \item \textbf{An opportunity-level bridge to auction theory.} We zoom into single opportunities. The winning fee splits into an approximately invariant floor and a contestable premium. The premium scales about linearly in the block time at high contention, and less steeply at low contention. The winner pays in full and the losers pay in part, an auction format between a winner-pays and an all-pay first-price auction.
\end{enumerate}

The primary observed causal effect is a large and robust compression of searcher bidding under flashblocks. However, the revert rate is not cleanly identified, so we report it descriptively rather than causally.

The rest of the paper is organized as follows. Section~\ref{sec:setting} covers the background and data, and Section~\ref{sec:model} the model and its predictions. Section~\ref{sec:methods} defines the outcomes, groups, flashblock reconstruction, and identification. Section~\ref{sec:results} reports the block structure, causal estimates, opportunity-level analysis, prediction tests, and regime stability. Section~\ref{sec:related} reviews related work, and Section~\ref{sec:conclusion} concludes.


\subfile{causal_v2}

\end{document}

%% file: causal_v2.tex
\section{Background}
\label{sec:setting}
In July 2025 Base introduced \emph{flashblocks}:
ten preconfirmation sub-blocks of 200~ms streamed within each 2~s block. The canonical block time stayed at 2~s. The intervention changes only the cadence of the priority gas auction (PGA): a single 2~s auction becomes ten sequential 200~ms sub-auctions. In this study, we compare June 2025 (pre) with August 2025 (post). We exclude July as the transition period. Two further post-intervention windows, October 2025 (three months after launch) and June 2026 (eleven months after launch), are used to test whether the August 2025 regime is transitional or the mature steady state (Appendix~\ref{sec:stability}).

\label{sec:data}
Transaction- and block-level data for June and August 2025 come from a Base archive node. The block data supply the base fee, from which we define the effective priority fee as $p = \texttt{gas\_price} - \texttt{base\_fee}$. We exclude system and deposit transactions ($\texttt{tx\_type} = 126$) from per-transaction rates and from the flashblock reconstruction. Appendix~\ref{app:fields} documents the extraction and fields. The later windows are processed identically.

Table~\ref{tab:data} summarizes the pre and the post months. Coverage is comparable: about $1.3$ million blocks and $293$--$296$ million transactions each, at roughly $224$ transactions per block, over $57.2$ million distinct sending addresses across the two months. Two aggregate shifts are already visible: the revert rate rises from $0.078$ to $0.114$ and the mean effective priority fee per gas falls from $0.038$ to $0.025$~gwei. These pooled movements mix searchers with everyone else and confound flashblocks with the July-to-August market shift. Our goal is primarily to isolate the part attributable to flashblocks.

\begin{table}[t]
  \centering
  \caption{Summary statistics by month. Per-transaction rates exclude system and deposit
  transactions.}
  \label{tab:data}
  \begin{tabular}{lrr}
    \toprule
    Statistic & June (pre) & August (post) \\
    \midrule
    Blocks                            & 1{,}294{,}626 & 1{,}333{,}894 \\
    Transactions                      & 292{,}876{,}642 & 295{,}937{,}009 \\
    Transactions per block            & 226 & 222 \\
    Revert rate                       & 0.078 & 0.114 \\
    Mean gas used                     & 275{,}332 & 227{,}206 \\
    Mean effective priority fee per gas (gwei) & 0.038 & 0.025 \\
    \bottomrule
  \end{tabular}
\end{table}

\section{Model and predictions}
\label{sec:model}
In this section, we connect the opportunity-level measures to two models of how searchers respond to flashblocks, so the aggregates of Section~\ref{sec:did} can be read as auction behavior.

The theoretical models support two flashblock-dependent predictions. First, arbitrage value per opportunity scales with block time. Under a Geometric Brownian Motion (GBM) with jumps for the underlying asset price, the arbitrage profit available on a pool over a window of length $B$ is given by
\begin{equation}
 V(B) \;=\; \kappa\,\sigma^2 B \;+\; J,
  \label{eq:sqrtt}
\end{equation}
where $\sigma$ is the price volatility, $\kappa$ a pool-specific constant, and $J$ a jump component that survives as $B \to 0$ because a discontinuous price move creates an arbitrage regardless of how short the window is,~\cite{bundi2026optimal,fritsch2024measuring}.

The first term is the diffusive part. Loss-versus-rebalancing (LVR) is a curvature cost, quadratic in the price move~\cite{milionis2022lvr}. Over a window of length $B$ the price moves by about $\sigma\sqrt{B}$. The value is the square of that move, so it scales as $\sigma^2 B$ (Appendix~\ref{app:blockvalue}). Nezlobin and Tassy~\cite{nezlobin2025lvr} derive the same loss for deterministic block times, the case that applies to Base.
Fritsch and Canidio~\cite{fritsch2024measuring} measure the block-time dependence empirically and report a 20 to 70\% reduction in arbitrage losses from 12~s to 100~ms. The decline is pair-dependent and flattens above the continuous limit. This supports the invariant term $J$ in \eqref{eq:sqrtt}. Reducing the flashblock length reduces $\kappa\sigma^2 B$ but not $J$. The winning fee sits below this value at low competition and rises toward it as more searchers compete. Its premium over the floor scales as $B^{\alpha}$, with $\alpha$ measured from about $0.5$ to $1.2$ across competitor counts (Section~\ref{sec:theory}).

Second, under the PGA a single opportunity is a sealed-bid first-price auction.
Suppose there are $n$ risk-neutral players and one item to buy. The item corresponds to the arbitrage opportunity. Each player draws a value independently from the same distribution function $F$, continuous on an interval $[\underline{v},\bar{v}]\in [0,\infty)$ with positive density $f$. Each player bids $b(v)$ on its value $v$. The winner is the highest bidder. The winner pays its bid. All losers pay a fraction $c$ of their bids. We solve for the symmetric equilibrium of this game, in which $b(\cdot)$ is increasing.

The block is then split into rounds. Each player arrives in the current PGA round with probability $p$ and in the next round with probability $1-p$. A player that moves to the next round loses to any player in the current round, regardless of bids. If all players move to the next round, they compete there on bids. Payments are unchanged: the winner pays its full bid and each loser pays the fraction $c$ of its bid.

\begin{proposition}\label{prop:equilibrium_bidding}
    Equilibrium bidding function is:
    \begin{equation*}
        b(v)=\frac{(n-1)\int_{\underline{v}}^{v}tf(t)[p^2(r+pF(t))^{n-2}+r^nF(t)^{n-2}]dt}{c+(1-c)[p(r+pF(v))^{n-1}+r^nF(v)^{n-1}]}.
    \end{equation*}
\end{proposition}

The numerator is the expected value won, and the denominator is the expected fraction of the bid paid. Loser payments therefore depress bids relative to a winner-pays auction. Both the value distribution $F$ and the current-round arrival probability $p$ vary with $B$. Longer windows admit larger price dislocations and more timely arrivals, so the model predicts lower revenue as $B$ falls. We find that the revenue of the auction is increasing in the probability of getting included in the current block, as long as the latter is larger than a threshold. Interestingly, the threshold does not depend on the value distribution function, but only on the number of bidders $n$.

\begin{proposition}\label{prop:revenue_monotonicity}
    The revenue of the auction is increasing whenever $p\geq \frac{n}{n+2}$.
\end{proposition}

However, point-wise bidding function $b(v)$ does not necessarily increase in $p$, as it enters both numerator and denominators of its calculation,~\eqref{bid_as_function_of_p}. Indeed, there are examples where bidding does not increase in increasing $p$. 
Appendix~\ref{app:proofs} contains the example described above and proofs of propositions in this section.

\section{Methods}
\label{sec:methods}

\subsection{Outcomes and identification}
\label{sec:estimand}

Let a transaction have effective priority fee per gas $p = \texttt{gas\_price} - \texttt{base\_fee}$, declared bid $b = \texttt{max\_priority\_fee\_per\_gas}$, gas used $g$, and a revert indicator. We study three families of outcomes $Y$:
\begin{itemize}
  \item \textbf{Bidding:} the effective priority fee $p$ and the declared bid $b$.
  \item \textbf{Contention:} the revert rate and the gas burned by reverted transactions
        (deadweight).
  \item \textbf{Fee concentration:} the block-level index $C$ of Section~\ref{sec:conc}, the
        share of within-block priority-fee value paid at the top of the block.
\end{itemize}

The unit of an analysis is an address $i$ on day $t$. Each address belongs to a group $G_i \in \{S, C\}$,
searchers ($S$) or a control group ($C$), and each day to a regime $R_t \in \{\text{pre},
\text{post}\}$. Only searchers in the post regime are treated. Writing $Y_{it}(1)$ and $Y_{it}(0)$
for the treated and untreated potential outcomes, the target is the average treatment effect on
the treated (ATT),
\begin{equation}
  \tau \;=\; \mathbb{E}\big[\,Y_{it}(1) - Y_{it}(0)\ \big|\ G_i = S,\ R_t = \text{post}\,\big].
  \label{eq:att}
\end{equation}

\label{sec:ident}
Base is a single chain treated at one instant, so a before/after difference confounds the
flashblock effect with the larger July-to-August market shift. Including a
block-time term in the regression does not fix this: block time takes one value per month, so it is collinear with a month indicator and its coefficient is the confounded June/August contrast.

We use a control group for identification. More specifically, searchers are the treated group and we use a group of
comparable non-competitive addresses (Section~\ref{sec:groups}) untreated by flashblocks as the control group. The
estimator is the interaction coefficient in a two-way fixed-effects regression on the address-day
panel,
\begin{equation}
  Y_{it} \;=\; \tau\,(S_i \cdot P_t) \;+\; \alpha_i \;+\; \delta_t \;+\; \varepsilon_{it},
  \label{eq:twfe}
\end{equation}
where $S_i = \mathbf{1}[G_i = S]$, $P_t = \mathbf{1}[R_t = \text{post}]$, $\alpha_i$ are address
effects, and $\delta_t$ are day effects. The address effects absorb $S_i$ and every
time-invariant difference between the groups; the day effects absorb $P_t$ and every shock common
to both groups (ETH price, volatility, base fee). 

The coefficient $\tau$ equals the estimand \eqref{eq:att} under the assumption that,
absent treatment, the two groups (treated and the control) would have moved together:
\begin{equation}
  \mathbb{E}\big[\,Y_{it}(0)\ \big|\ S,\text{post}\,\big] - \mathbb{E}\big[\,Y_{it}(0)\ \big|\ S,\text{pre}\,\big]
  \;=\;
  \mathbb{E}\big[\,Y_{it}(0)\ \big|\ C,\text{post}\,\big] - \mathbb{E}\big[\,Y_{it}(0)\ \big|\ C,\text{pre}\,\big].
  \label{eq:pt}
\end{equation}
The above condition is on \emph{changes}: $\alpha_i$ already removes any fixed gap between the
groups, so only their untreated pre-to-post movements must match. The condition fails if a shock hits
the groups differentially across the two regimes. We check it with two auxiliary regressions. An event
study interacts the searcher indicator with each day relative to the last pre-period day; \eqref{eq:pt}
requires the pre-period interactions to vanish, and the post-period interactions average to $\tau$
(Figure~\ref{fig:did-eff}). A linear pre-trend regression on the pre-period days alone requires a zero
searcher-by-day slope. A nonzero value in either means the two groups were already diverging before
treatment.

The control is matched on activity, not on fees, and pays far below the searchers
(Table~\ref{tab:groups}). This gap is what we study, not a flaw, since searchers are the ones
paying the top-of-block premium. The design needs the two groups' fees to change by the same amount
absent flashblocks, not to start at the same level. The risk from a large gap is the choice between an
equal-amount and an equal-percentage comparison. A shock that moves fees by a fixed percentage would
move the searcher fee much more than the near-zero control fee, and the day effects, which subtract a
fixed amount, would not remove it. Two things limit this. The fees step down at the July launch. They do not drift apart before it
(Figure~\ref{fig:did-eff}). June and August had comparable volatility, $66\%$ against $70\%$ annualized.
A volatility gap between the two months is therefore not driving the estimate. The near-zero control fee
still makes an equal-percentage trend hard to check.

\subsection{Groups and classification}
\label{sec:groups}

Addresses are labeled from behavior aggregated over both months. The classification is strict. A looser top-percentile cut would mix genuine searchers with ordinary high-frequency addresses.
Both the searchers and the control group clear a frequency bar at the 99.9th percentile of the per-address transaction count over the two months
($n^\star = 405$ transactions), and searchers must be competitive on \emph{both} markers: a revert
rate at or above $r^\star = 0.20$ \emph{and} a maximum declared priority fee at or above the 95th
percentile ($b^\star = 0.011$~gwei). We only consider addresses active in both months.\footnote{Temporary addresses used only briefly cannot meet both restrictions: 405 transactions and activity in both months. So, they do not enter the searcher or control groups.} The groups are
\begin{align*}
  \text{searcher} \;&:\; n \ge n^\star \ \text{and}\ \big(r \ge r^\star \ \text{and}\ b_{\max} \ge b^\star\big), \\
  \text{control (activity-matched)} \;&:\; n \ge n^\star \ \text{and}\ r < r^\star \ \text{and}\ b_{\max} < b^\star, \\
  \text{retail (placebo)} \;&:\; n < 45 \ \text{and}\ r < 0.05.
\end{align*}
Searchers and activity-matched controls clear the same activity threshold but differ on the two
competitive markers. Requiring both markers isolates a competitive core at the cost of coverage.
Because the screen conditions on revert rate and maximum bid, the estimates apply to this selected
core rather than to a population fixed before treatment. Flat pre-trends do not rule out selection,
but rescreening on June alone recovers the same fee and bid compression (Section~\ref{sec:did}).

The disjoint retail placebo contains addresses with fewer than 45 transactions and a revert rate
below 0.05. These low-frequency addresses should respond little to flashblocks. Table~\ref{tab:groups}
reports the group counts and means. The retail mean priority fee exceeds that of the activity-matched
control because the control is also screened on bidding.

\begin{table}[t]
  \centering
  \caption{Address groups: counts and behavioral means over both months.}
  \label{tab:groups}
  \begin{tabular}{lrrr}
    \toprule
    Group & Addresses & Mean revert rate & Mean priority fee (gwei) \\
    \midrule
    Searcher                 & 3{,}032   & 0.748 & 0.127 \\
    Control (activity-matched) & 8{,}053 & 0.028 & 0.001 \\
    Retail (placebo)         & 1{,}322{,}706 & 0.001 & 0.021 \\
    \bottomrule
  \end{tabular}
\end{table}

\subsection{Flashblock reconstruction}
\label{sec:recon}

Flashblocks are not distinct on-chain blocks, so we recover the 200~ms structure from within-block
fee ordering. During each window, the sequencer draws from a descending-fee queue while continuing
to accept new arrivals. The fee sequence therefore forms descending runs interrupted by upward jumps.
Order a block's transactions by execution position and let $p_i$ and $s_i$ be the effective priority fee and sender of the $i$-th. An \emph{upward jump} is a pair with $p_i > p_{i-1}$. We discard same-sender jumps ($s_i = s_{i-1}$) as they are mechanical and forced by nonce ordering. A \emph{cross-sender} upward jump ($s_i \neq s_{i-1}, p_i > p_{i-1}$) cannot arise from one static price-and-nonce heap, which always pops the highest-fee ready transaction. It marks an ordering reset. A reset is either a flashblock boundary or a mid-window live preemption, where a late higher-fee arrival is served ahead of the remaining queue (Appendix~\ref{sec:livepull}). We count resets, so the slice count approximates rather than bounds the number of flashblock boundaries. The live-pull appendix separates the two. As a falsification, the pre-flashblock month must collapse to one slice per block.

\subsection{Fee-concentration index}
\label{sec:conc}

Let transaction $j$ pay priority-fee value $\phi_j = p_j g_j$ and hold cumulative-gas position
$\gamma_j = G_b^{-1}\sum_{k \le j} g_k$, where $G_b$ is the block's total gas. The index is
\begin{equation}
  C_b \;=\; \frac{\sum_{j:\,\gamma_j \le 0.1}\phi_j}{\sum_j \phi_j},
  \label{eq:conc}
\end{equation}
the share of within-block priority-fee value paid in the first gas-decile (the block top, in execution order). Because the cutoff is a gas share and $\phi_j$ is gas-weighted, the gas factors
cancel and
\begin{equation}
  10\,C_b \;\approx\; \frac{\bar p_{\text{top}}}{\bar p_{\text{block}}},
  \label{eq:conc-rate}
\end{equation}
the ratio of the gas-weighted mean fee rate in the top decile to the block-wide mean rate. $C_b$
therefore measures fee-\emph{rate} concentration: a large transaction raises it only through its
rate, and one paying the block-average rate leaves $C_b$ at its neutral value $0.1$. A single
auction seats the highest rates at the top, $C_b \to 1$; ten sub-auctions recur high rates through
the block, $\bar p_{\text{top}} \to \bar p_{\text{block}}$ and $C_b \to 0.1$. If every transaction paid the same
per-gas rate, $C_b = 0.1$. Concentrating the highest rates at the block top drives it toward $1$.

\section{Results}
\label{sec:results}

\subsection{Block structure}
\label{sec:micro}

The reconstruction of Section~\ref{sec:recon} confirms the regime change directly.
Pre-flashblocks, every block is a single
fee-sorted slice with no cross-sender resets; post-flashblocks, blocks fragment into a mean of
$19.5$ inferred slices, $94.4\%$ of upward fee jumps are cross-sender, and fee concentration $C$
falls from $0.89$ to $0.27$, with $98.3\%$ of blocks exceeding ten inferred slices post-flashblocks
against a single slice pre-flashblocks. Figure~\ref{fig:fee} shows the fee-position distribution. The
slice count exceeds ten because ordering locks incrementally within each 200~ms window, so it measures ordering fragmentation rather than an exact sub-block count. Whether this reflects searchers timing their submissions within the window is addressed in Appendix~\ref{sec:hyp}. The builder also preempts the sorted queue mid-window, a live priority-queue pull documented in Appendix~\ref{sec:livepull}.

The redistribution of fees persists in the later post-launch months (Figure~\ref{fig:fee}).
Across all senders, fee concentration $C$ is $0.27$ in October 2025 and $0.36$ in June 2026,
compared with $0.89$ before flashblocks and $0.27$ in August 2025.
Sequencing continues to change within this regime. Mean inferred slices per block decline
from $19.5$ in August to $12.0$ in October and $10.0$ in June 2026.
The share of cross-sender jumps classified as returns falls from $33.2\%$ to $19.8\%$ and $0.4\%$,
respectively. These measurements establish persistence in block structure.
Comparisons of searcher fee levels across months remain confounded by market conditions and
changes in the screened population, so they do not identify a long-run causal effect.
Appendix~\ref{sec:stability} reports the detailed comparison.

\begin{figure}[t]
  \centering
  \includegraphics[width=0.80\textwidth]{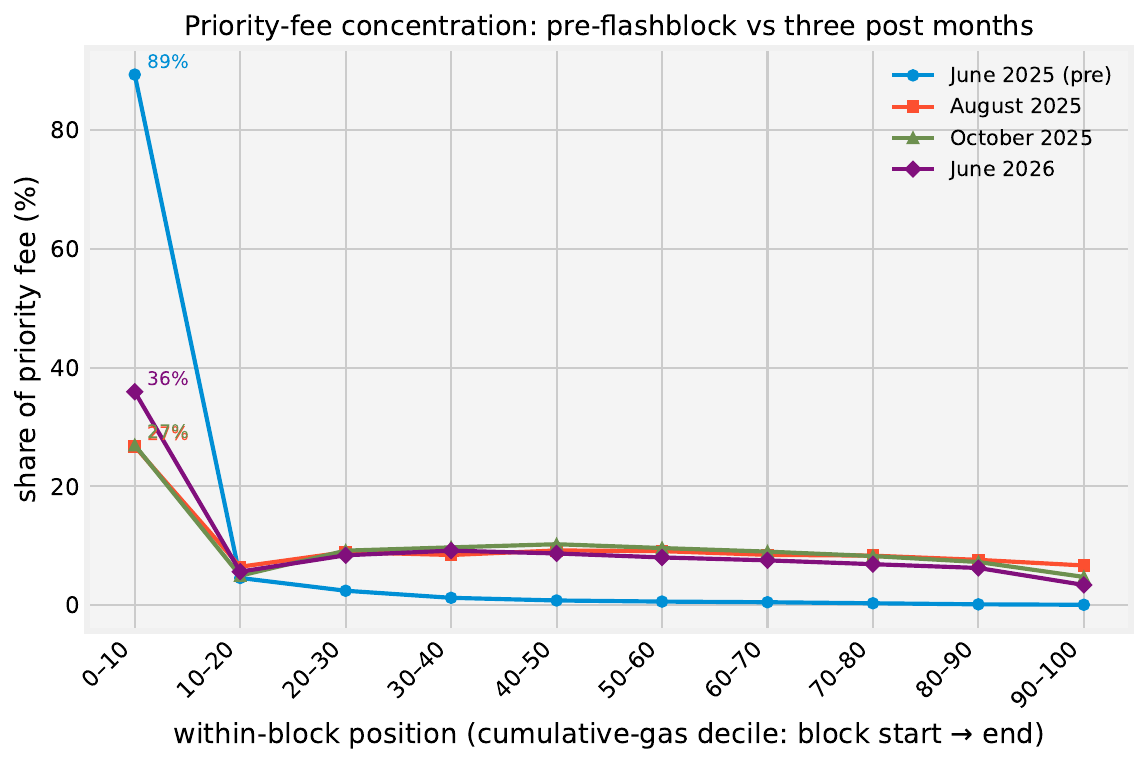}
  \caption{Share of within-block priority-fee value by cumulative-gas decile in June 2025,
  August 2025, October 2025, and June 2026. The pre-flashblock concentration at the block top
  remains dispersed across the block in all three post-launch months.}
  \label{fig:fee}
\end{figure}

\subsection{Causal estimates}
\label{sec:did}

Equation~\eqref{eq:twfe} is estimated on the address-day panel of $3{,}032$ searchers and
$8{,}053$ activity-matched controls over $61$ days (Section~\ref{sec:groups}).
Table~\ref{tab:did} reports the within-group pre-to-post changes, fixed-effects estimates
with address-clustered standard errors, and June pre-trend tests.\footnote{The DiD column reports the fixed-effects estimate from
Equation~\eqref{eq:twfe}. Panel weighting can make it differ from the raw difference
between the two change columns (Appendix~\ref{app:robustness}).}

The effective priority fee falls by $0.187$~gwei more for searchers than controls,
equivalent to $59\%$ of the pre-period mean of $0.317$~gwei for the competitive core
active in both months.\footnote{Fee levels and concentration measures use different
populations across analyses and are not directly comparable (Appendix~\ref{app:robustness}).}
The declared bid falls by $0.228$~gwei more. The June pre-trend $p$-values are $0.79$
and $0.98$, respectively. With an inclusion opportunity every 200~ms, a transaction no longer outbids the whole block to reach the
top, so realized fees and the bids that set them decline for the population that paid the
top-of-block premium. Fee concentration gives the largest effect: $C$ falls from $0.98$ to $0.28$
for searchers against $0.07$ to $0.09$ for controls, $\hat\tau = -0.72$. Its June pre-trend is flat ($p = 0.45$), consistent with parallel trends, though it is measured on a coarse group-hour panel where the test has low power. The per-group $C$ takes the block's top gas-decile as the cut, then measures the share of that group's within-block priority-fee value falling in it.

Figure~\ref{fig:did-eff} shows the fee event study: each point is the searcher-minus-control fee
gap on a day, relative to the last pre-period day. The gap is flat and centered on zero through the
pre-period ($p = 0.79$), then steps down across the post-period, consistent with a level shift at
the July intervention rather than a pre-existing divergence.

The outcomes revert rate, overbid, and transaction count are not cleanly identified. The revert-rate
estimate is positive ($\hat\tau = +0.021$) and significant, but the strict searcher classification conditions
on a high revert rate, so the treated group is selected on that very outcome and its pre-period
differential slope is nonzero ($p < 0.001$); The value in equation~\eqref{eq:pt} fails and the coefficient is descriptive, not causal. Overbid and transaction count likewise carry nonzero pre-trends. Figure~\ref{fig:did-revert} makes this visible: the revert-rate gap scatters widely through the pre-period rather than holding constant.

\begin{table}[t]
  \centering
  \caption{Difference-in-differences against the activity-matched control. $^{*}$ marks the
  difference-in-differences estimate significant at $p < 0.01$ (from its address-clustered standard   error); the separate June pre-trend $p$-value, for which a large value ($> 0.05$) supports the identification assumption \eqref{eq:pt}. Fee outcomes in gwei. Overbid, revert rate, and transactions per address-day have nonzero pre-trends ($p < 0.001$). Their estimates are descriptive, not causal.}
  \label{tab:did}
  \input{did_results}
\end{table}

\begin{figure}[t]
  \centering
  \includegraphics[width=0.86\textwidth]{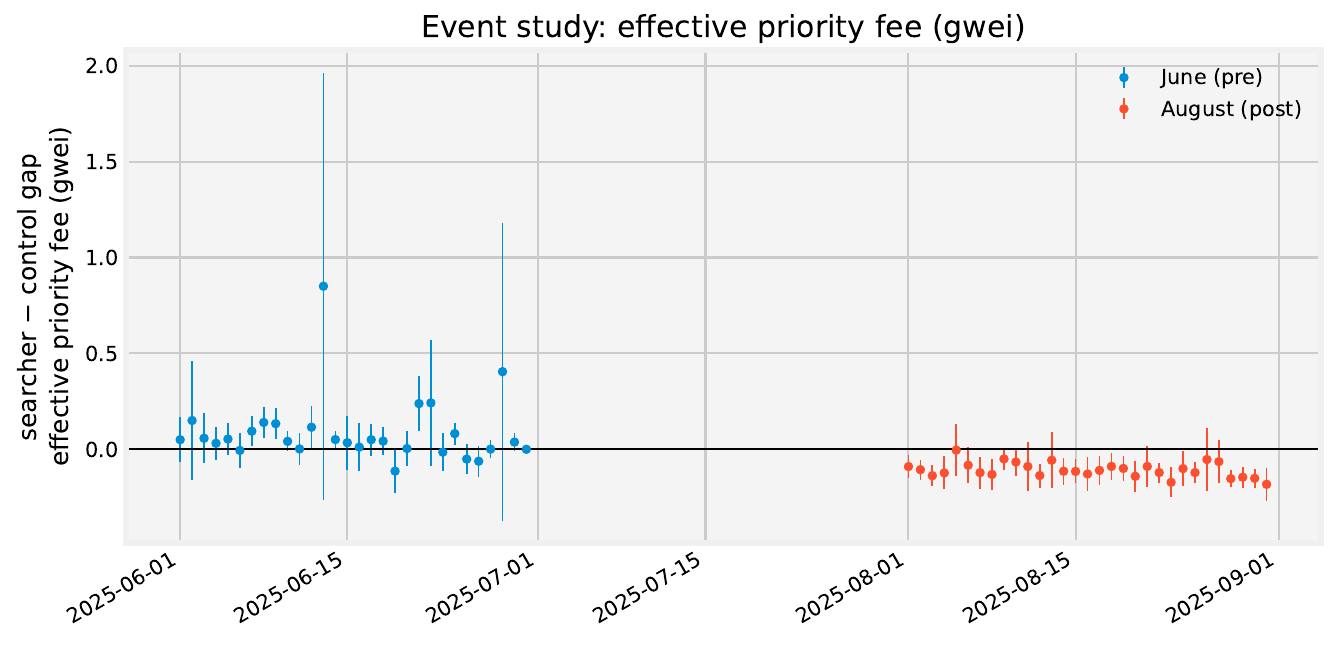}
  \caption{Event study for the effective priority fee: searcher-minus-control gap by day,
  relative to the last pre-period day, with 95\% intervals. The pre-period gap is roughly stable
  and closes across the post-period.}
  \label{fig:did-eff}
\end{figure}

The searcher screen uses revert rate and top bids over both months. So an address can enter the treated group on its post-flashblock behavior. We rescreen on June alone. The thresholds are the June 99.9th frequency percentile ($339$ transactions) and the 95th bid percentile ($0.006$~gwei), giving $2{,}436$ searchers and $4{,}158$ controls. The effective fee drops by $0.120$~gwei and the bid by $0.165$~gwei, both significant with flat pre-trends, $64\%$ and $72\%$ of the main estimates. This removes selection on post-period bids. It does not remove mean reversion. The treated group is picked on high June bids and the control on low ones, so the two drift together on their own. A flat pre-trend slope does not rule this out. 

The fee and concentration estimates retain their signs under the alternative retail control (Appendix~\ref{app:robustness}).

\subsection{Opportunity-level competition and late arrival}
\label{sec:oppmatch}

The address-day results in Section~\ref{sec:did} measure fee changes relative to controls. We next compare bids and execution order among attempts at the same arbitrage opportunity.
Attempts are matched by the underlying liquidity pool, since contract addresses and method selectors do not uniquely identify competitors.

Successful trades reveal pools through decoded \texttt{Swap} events. An opportunity is a pool dislocation delimited by a successful capture. Several opportunities can occur within one 2~s block.
Reverted transactions are assigned pools using their contract's successful trades.
This attribution covers $96\%$ of core reverted fee value in August 2025.
Appendix~\ref{app:oppmatch} provides the implementation details.

In August 2025, $766{,}510$ captures have at least one linked reverting attempt.
A higher-fee attempt reverts one reconstructed slice after its linked winner in $5.4$--$6.6\%$ of these contested captures. The range spans dominant-pool and proportional attribution rules. 
Weighting by the winner's realized priority-fee value gives $1.9$--$2.7\%$.

These rates are upper bounds on latency losses. Reverts can arise from slippage or stale opportunities, and pool attribution can misassign multi-pool searchers.
The $96\%$ coverage figure does not bound attribution error.
Reconstructed slices can split a single flashblock (Section~\ref{sec:recon}).
Appendix~\ref{app:oppmatch} reports the two-slice comparison, bid gaps, and pre-period check.
Appendix~\ref{sec:hyp} reports additional transaction-level tests.

\subsection{Testing the predictions}
\label{sec:theory}
At the opportunity level, let $c$ denote the number of distinct core searchers contesting a pool within a 2~s block. This pool-block unit differs from the single captures used in the late-arrival analysis of Section~\ref{sec:oppmatch}. Competition is thin: $81.4\%$ of opportunities in June and $78.0\%$ in August contain one searcher, while the respective means are $1.41$ and $1.56$. Captures per pool-block increase from $1.07$ to $1.22$. Appendix~\ref{app:oppmatch} reports the full distribution. We therefore model each contest as an auction with a small number of participants.

Figure~\ref{fig:betac} plots the winning fee $\beta(c)$ in both regimes. The uncontested floor is similar across regimes, at $0.177$~gwei in June and $0.164$~gwei in August. We treat it as approximately invariant because the $7.9\%$ difference is not tested. The premium above this floor changes more sharply with auction duration. With comparable monthly volatility, the implied block-time exponent rises from $0.51$ at two searchers to $0.96$--$0.97$ at four and five searchers, and $1.12$ beyond five. The premium therefore approaches linear scaling in $B$ as competition increases and is closer to $\sqrt{B}$ under low contention, consistent with the value scaling derived in Appendix~\ref{app:blockvalue}.

The premium is also hard to separate from opportunity size, since contests with more searchers tend to be more valuable, so $\beta(c)$ partly measures value and not just competition. The check in Appendix~\ref{sec:routeb} reduces this but uses gas as a rough proxy for value. That check covers only the 200~ms data, so the cross-regime ratio above uses the raw $\beta(c)$. The longer 2~s window lets the price dislocate further, so the largest opportunities draw the most searchers. Latency also thins the field at 200~ms, where some searchers cannot compete.
Note that the fee is a lower bound on the arbitrage value, so this is a consistency check on \eqref{eq:sqrtt} and not a direct measurement of arbitrage value.

Partial loser payments are consistent with the first-price model. Reverting losers account for $11.4\%$ of total fee value on contested captures, and the median loser pays $0.15$ times the winner's fee value. Appendix~\ref{app:oppmatch} describes the matching and gas accounting behind these measures. Because we do not fit the equilibrium bid in \eqref{eq:bid} to $\beta(c)$, this evidence is qualitative.

The diffusion-plus-jump decomposition also predicts diminishing returns below 200~ms: the diffusion premium and within-window competition shrink with $B$, while the jump term $J$ remains.
Section~\ref{sec:shorter} quantifies this projection.

\begin{figure}[t]
  \centering
  \includegraphics[width=0.92\textwidth]{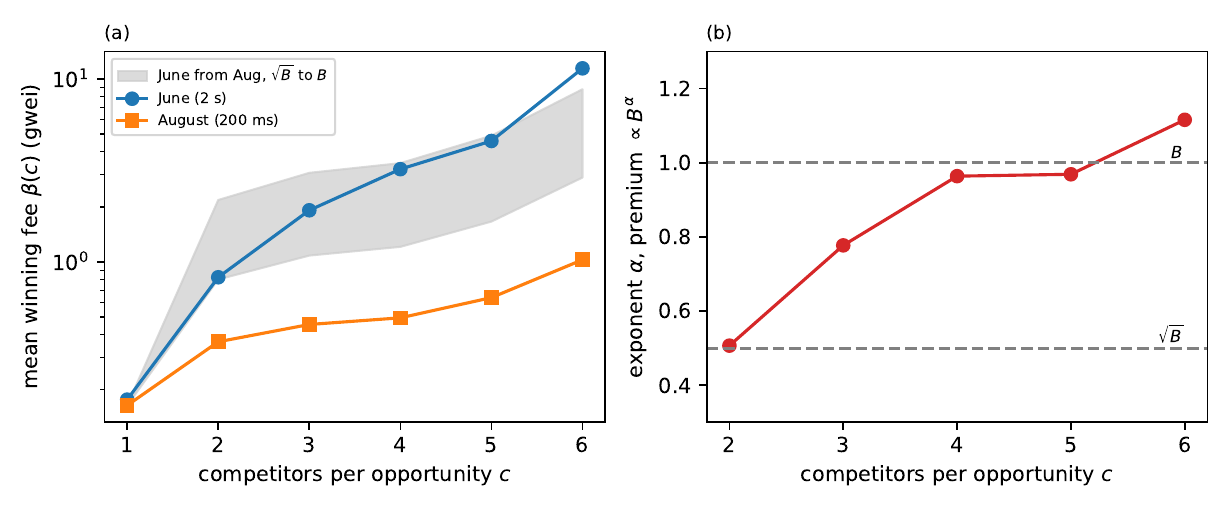}
  \caption{Winning fee premium by competitor count in the two regimes. (a) Mean winning
  fee $\beta(c)$ on a log scale; the floor $\beta(1)$ is nearly the same across regimes. The shaded band is the
  June premium implied from the August premium under scalings between $\sqrt{B}$ and $B$. June sits at
  the $\sqrt{B}$ edge at two searchers, crosses to the $B$ edge by four to five, and exits above the
  band at $c \ge 6$. (b) The implied block-time exponent $\alpha$ of the premium, with premium
  $\propto B^{\alpha}$, climbs from $0.51$ to $1.12$ and crosses the linear-$B$ line between four and five
  searchers.}
  \label{fig:betac}
\end{figure}

\subsection{Shorter flashblock windows}
\label{sec:shorter}
A natural policy question for the rollup operator is whether reducing the flashblock below 200~ms would compress searcher bids further. We answer it with a micro-auction model, calibrated on the two observed cadences (200~ms and 2~s) and on the competition measured per opportunity, and defer the construction to Appendices~\ref{app:extrap}
and~\ref{sec:routeb}. The projection points to a diminishing return, as the diffusive value falls with $B$ toward the invariant floor (Appendix~\ref{app:blockvalue}). At 200~ms the expected number of rival searchers in a single window is already only $0.14$, so most windows are uncontested and the winning bid sits at a value-capture floor. How large the drop is depends on the model, and the two differ several fold.
The model built from the measured competition (Appendix~\ref{sec:routeb}) has the winning fee fall by under $10\%$ between 200 and 50~ms, from $0.181$ to $0.169$~gwei. A simpler model
(Appendix~\ref{app:extrap}) gives a $33$ to $68\%$ drop. They differ on whether $\beta(c)$ stays the same at shorter windows: the simpler model only spreads the competition thinner, while Route B reuses the measured $\beta(c)$. Neither can be checked below 200~ms. No chain runs there yet, and the $2$~s and 200~ms bids show $\beta(c)$ changes with the window (Table~\ref{tab:betac} in Appendix~\ref{app:oppmatch}), so both are rough bounds, not predictions.{Note that reducing block time to very little number distorts $p,(1-p)$ assumption for the current round/next round probabilities. In fact, the ordering policy converges to FCFS if block time converges to 0. However, the shortest block time considered in the analysis is $50$~ms, which makes sure that the transaction either lands in the current round or the next round with certainty. Typical searcher latency is in the range $[10,40]$.}

\section{Related work}
\label{sec:related}

This work draws on four lines of research: transaction-fee mechanism design, maximal extractable value (MEV) and priority gas auctions, preconfirmations and faster block times, and the block-time dependence of arbitrage.

\textbf{Transaction-fee mechanisms.} Roughgarden~\cite{roughgarden2020eip1559} analyzes EIP-1559 and introduces incentive-compatibility notions for the blockchain setting. Leonardos et al.~\cite{leonardos2021dynamical} study the base-fee update rule as a dynamical system and bound the step size that guarantees convergence, and Reijsbergen et al.~\cite{reijsbergen2021honeymoon} measure the block-size fluctuations observed one month after the launch of EIP-1559. These works are primarily related to the base fee, which flashblocks leave unchanged. We study the priority fee and how it responds when the auction cadence changes.

\textbf{MEV and priority gas auctions.} Daian et al.~\cite{daian2020flashboys} document priority gas auctions among arbitrage bots and risks they introduce to blockchain consensus, and Qin et al.~\cite{qin2022quantifying} quantify extractable value from arbitrage, liquidations, and sandwiches. Wang et al.~\cite{wang2026blockspace} measure spam MEV on Base and Arbitrum, where searchers submit speculative transactions resolved only at execution. These works study ordering value at the level of the whole block. We measure how the same competition changes when the block is split into ten sequential sub-auctions.

\textbf{Preconfirmations and faster blocks.} Flashblocks streams 200~ms preconfirmation sub-blocks within Base's 2~s block~\cite{flashblocks2025}, a sequencer-issued form of the preconfirmations proposed for rollups~\cite{drake2024preconf}. Schwarz-Schilling et al.~\cite{schwarzschilling2023timing} model the timing games proposers play to capture MEV within a slot, and Fritsch et al.~\cite{fritsch2025mev} analyze arbitrage under an auctioned time advantage. These works price the value of time and sequencing. In market design, Budish et al.~\cite{budish2015hft} show that a continuous order book drives a latency arms race, which batching into frequent discrete auctions removes. Flashblocks instead add ten sequential races within the block. We measure the net effect on searcher bids from on-chain data.

\textbf{Block-time arbitrage.} Milionis et
al.~\cite{milionis2022lvr,milionis2023fees} formalize loss-versus-rebalancing and its extension to
trading fees, while Fritsch and Canidio~\cite{fritsch2024measuring} measure its block-time dependence.
Nezlobin and Tassy~\cite{nezlobin2025lvr} derive the dependence for deterministic and general block
times, and Bundi~\cite{bundi2026optimal} decomposes it into diffusion and invariant jump channels. We
use the diffusion-plus-jump decomposition to connect arbitrage value to the priority fees searchers
pay under flashblocks.

\section{Conclusion}
\label{sec:conclusion}
Flashblocks compress searcher bidding on Base. Restructuring the priority gas auction into ten 200~ms sub-auctions lowers the searcher effective priority fee by $0.187$~gwei ($0.168$ before panel weighting) relative to a matched control, $59\%$ of its pre-period level, and collapses the top-of-block fee concentration from $0.98$ to $0.28$. The effect is a level shift at the July intervention. The block-level regime persists eleven months after launch, though we do not re-estimate the causal effect on the later months. At the opportunity level, ordering within a block is decided partly by latency. A higher-bidding transaction loses to a one-flashblock deficit in $1.9$ to $2.7\%$ of contested opportunities weighted by value, and $5.4$ to $6.6\%$ unweighted, both upper bounds. The compression is first-order at 200~ms and diminishes below it. The thinning model of Appendix~\ref{app:extrap} and the measured-competition model of Appendix~\ref{sec:routeb} bound the drop. Two questions remain open. The revert-rate response is not identified under our screen and is reported descriptively. The shorter-window projection is a model extrapolation rather than a measurement, because no sub-200~ms regime exists on-chain to date.

\bibliographystyle{splncs04}
\bibliography{references}

\appendix

\section*{Contents of the Appendices}
{\small
\noindent\begin{tabular}{@{}r@{\quad}p{0.82\textwidth}@{\quad}r@{}}
\ref{app:proofs}      & \nameref{app:proofs}      & \pageref{app:proofs}      \\[2pt]
\ref{app:blockvalue}  & \nameref{app:blockvalue}  & \pageref{app:blockvalue}  \\[2pt]
\ref{app:fields}      & \nameref{app:fields}      & \pageref{app:fields}      \\[2pt]
\ref{sec:livepull}    & \nameref{sec:livepull}    & \pageref{sec:livepull}    \\[2pt]
\ref{sec:stability}   & \nameref{sec:stability}   & \pageref{sec:stability}   \\[2pt]
\ref{app:robustness}  & \nameref{app:robustness}  & \pageref{app:robustness}  \\[2pt]
\ref{app:oppmatch}    & \nameref{app:oppmatch}    & \pageref{app:oppmatch}    \\[2pt]
\ref{sec:hyp}         & \nameref{sec:hyp}         & \pageref{sec:hyp}         \\[2pt]
\ref{sec:vol}         & \nameref{sec:vol}         & \pageref{sec:vol}         \\[2pt]
\ref{app:extrap}      & \nameref{app:extrap}      & \pageref{app:extrap}      \\[2pt]
\ref{sec:routeb}      & \nameref{sec:routeb}      & \pageref{sec:routeb}      \\
\end{tabular}}

\bigskip

\section{Proofs}
\label{app:proofs}

\begin{proof}[of Proposition~\ref{prop:equilibrium_bidding}]
    Suppose a player gets value sample $x$, and bids $b(x)$. If this player makes in the current round, with probability $p$, all other players must either make it in the next round, or bid less than this player, this happens with probability $pF(x)$. Let $r:=1-p$. That is, this player wins with probability $p(r+pF(x))^{n-1}$.
    If the original player makes it to the next round, the only chance it can win is if all others make it in the next round. This contributes probability $r^nF(x)^{n-1}$.
    The winning probability is

    \begin{equation}\label{prob_of_winning_x}
        Q(x)=p[r+pF(x)]^{n-1}+r^nF(x)^{n-1}.
    \end{equation}
    Its derivative is equal to:

    \begin{equation}\label{deriv_prob_of_winning}
        Q'(x)=(n-1)f(x)[p^2[r+pF(x)]^{n-2} + r^nF(x)^{n-2}].
    \end{equation}
    
    A bidder with true value $v$, bidding as if it had value $x$ has the expected utility equal to:

    \begin{equation*}
        U(v,x)=vQ(x)-(c+(1-c)Q(x))b(x).
    \end{equation*}

    The first order condition at $x=v$ is equal to:

    \begin{equation}\label{differential_equation}
        \frac{d}{dv}(c+(1-c)Q(v))b(v) = vQ'(v).
    \end{equation}
    Let $q_0:=Q(\underline{v})=pr^{n-1}$. The boundary condition is: 

    \begin{equation*}
    [c+(1-c)q_0]b(\underline{v}) = 0.
    \end{equation*}

    Integrating~\eqref{differential_equation} gives:

    \begin{equation*}
        b(v) = \frac{\int_{\underline{v}}^v tQ'(t)dt}{c+(1-c)Q(v)}.
    \end{equation*}

    Plugging in the formula for $Q'(t)$ gives the bidding function from the proposition.\qed
\end{proof}

\begin{proof}[of Proposition~\ref{prop:revenue_monotonicity}]
    Let $q(v;p):=p(1-p+pF(v))^{n-1}+(1-p)^n F(v)^{n-1}$ be the equilibrium allocation probability for a bidder with value $v$. Define $q_v(v;p):=\frac{\partial q(v;p)}{\partial v}$. Then, $q_v(v;p)=(n-1)f(v)\left[p^2(1-p+pF(v))^{n-2} + (1-p)^nF(v)^{n-2}\right]$.

    Symmetric equilibrium bidding function is: 

    \begin{equation}\label{bid_as_function_of_p}
        b(v;p) = \frac{\int_{\underline{v}}^v t q_v(t;p)dt}{c+(1-c)q(v;p)}.
    \end{equation}

    The interim expected payment of a bidder of value $v$ is: 

    \begin{equation}
        m(v;p) = \int_{\underline{v}}^{v}tq_v(t,p)dt.
    \end{equation}
    Expected total auction revenue is: 

    \begin{equation}
        R(p) = n\int_{\underline{v}}^{\bar v}v(1-F(v))q_v(v,p)dv.
    \end{equation}

    We can show that if $p^*=\frac{n}{n+2}$, for any $p>p^*$, $R'(p)>0$, hence the revenue is increasing in $p$ for sufficiently large $p$. 
    Let $x=F(v)$, $V(x)=F^{-1}(x)$. Then 
    
    $$R(p)=n(n-1)\int_{0}^{1}V(x)(1-x)[p^2(1-p+px)^{n-2}+(1-p)^nx^{n-2}]dx.$$
    Define $A(p,x):=p^2(1-p+px)^{n-2}+(1-p)^nx^{n-2}$. Its derivative with respect to $p$ is equal to $A_p(p,x)=p(1-p+px)^{n-3}(2-np(1-x))-n(1-p)^{n-1}x^{n-2}$. Then, $R'(p)=n(n-1)\int_{0}^{1}V(x)(1-x)A_p(p,x)dx.$

    Set $h_p(x):=(1-x)A_p(p,x)$. Since $V(x)$ is increasing, it is enough to show that the following integral is non-negative: 

    $$H_p(z)=\int_{z}^{1}h_p(x)dx,$$
    for any $z\in [0,1]$. 
    Let $Q:=1-z$ and $q:=1-x$, then

    \begin{equation}
        H_p(z)=pQ^2(1-pQ)^{n-2}-n(1-p)^{n-1}\int_{0}^{Q}q(1-q)^{n-2}dq.
    \end{equation}
    The following two inequalities hold: 
    $$\int_{0}^{Q}q(1-q)^{n-2}q\leq \frac{Q^2}{2},$$
    and $1-pQ\geq 1-p$.
    Therefore, 

    $$H_p(z)\geq pQ^2(1-p)^{n-2}-\frac{n}{2}(1-p)^{n-1}Q^2=Q^2 (1-p)^{n-2}[p-\frac{n}{2}(1-p)].$$

    Hence, $H_p(z)\geq$ is equivalent to $p\geq \frac{n}{2}(1-p)$, which is equivalent to $p\geq \frac{n}{n+2}$.

    Integration by parts gives: 

    $$\int_0^1 V(x)h_p(x)dx = V(0)H_p(0)+\int_0^1H_p(x)dV(x).$$
    Since $V$ is increasing and nonnegative, $dV\geq 0$, which gives that $H_p(x)\geq 0$ for every $x$ implies $R'(p)\geq 0$. This finishes the proof.
    \qed
\end{proof}

\begin{example}\label{example_nonmonotonicity_bid}
    Consider $n=3$, $F(v)=v$ for $v\in [0,1]$. Then, $q(v;p)=p(1-p+pv)^2+(1-p)^3v^2$. The equilibrium bid is $$b(v;p)=\frac{M(v;p)}{c+(1-c)q(v;p)}.$$
    
By direct integration we have: 
\begin{equation*}
    M(v;p)=\frac{v^3}{3}[3p^2(1-p)+2v(1-3p+3p^2)].
\end{equation*}
At $p=1$, $b(v;1)=\frac{2/3v^3}{c+(1-c)v^2}$. Differentiating the bid with respect to $p$ and evaluating it in point $p=1$ gives:
\begin{equation*}
    \frac{\partial b(v;p)}{\partial p}|_{p=1}=\frac{v^2[(1-c)v^2+6cv-3c]}{3[c+(1-c)v^2]^2}.
\end{equation*}
Its sign is exactly the sign of $(1-c)v^2+6cv-3c$. For every $c>0$, this function in $v$ is negative for sufficiently small $v$, as $\lim_{v\rightarrow 0+}[(1-c)v^2+6cv-3c]=-3c<0$. By continuity, for sufficiently large $p$, we do not have monotonicity in bids. 
\end{example}

\section{Block-time dependence of arbitrage value}
\label{app:blockvalue}
This appendix derives how the arbitrage value on a pool scales with the block time $B$, the diffusive term of \eqref{eq:sqrtt}. The value per opportunity grows linearly in $B$, whereas the price grows as $\sqrt{B}$. We follow the loss-versus-rebalancing analysis of Milionis, Moallemi, and Roughgarden~\cite{milionis2022lvr} and its block-time extensions~\cite{nezlobin2025lvr,bundi2026optimal}.

\paragraph{Curvature cost in continuous time.}
Let $P$ be the external price and $V(P)$ the value of the pool once it is arbitraged to $P$. Based on Envelope theorem, $V'(P)$ equals the pool's risky holding. Now, for a price move $\mathrm{d}P$ the pool value changes by
\begin{equation}
  \mathrm{d}V = V'(P)\,\mathrm{d}P + \tfrac{1}{2}V''(P)\,(\mathrm{d}P)^2,
  \label{eq:ito}
\end{equation}
where the first term is the benchmark wealth change, the pool assets traded at the external price without friction. The loss against that benchmark is the negative of the second term, since $V'' \le 0$.
Under a geometric Brownian price $\mathrm{d}P = \sigma P\,\mathrm{d}W$ the quadratic variation gives $(\mathrm{d}P)^2 = \sigma^2 P^2\,\mathrm{d}t$, so the loss accrues at the rate
\begin{equation}
  \frac{\mathrm{d}\,\mathrm{LVR}}{\mathrm{d}t} = \tfrac{1}{2}\Gamma(P)\,\sigma^2 P^2,
  \label{eq:lvrrate}
\end{equation}
where $\Gamma = -V'' \ge 0$ is the pool curvature. Note that this rate is proportional to $\sigma^2$ and doesn't depend on the block time.

\paragraph{Loss per arbitrage event.}
With discrete blocks the pool is corrected once per block. Moreover, the pool fee opens a no-arbitrage band of full width $\gamma$, and an arbitrage event occurs when the price leaves the band by an overshoot $d$.
Consider a pool with depth $\ell$. The arbitrageur's marginal profit at the mispricing $u$ is $l \cdot u$.
The total profit (or, the loss for the LP) during this event is

\begin{equation}
  L_{\mathrm{event}} = \int_0^{d}\ell\,u\,\mathrm{d}u = \tfrac{1}{2}\ell\,d^2,
  \label{eq:levent}
\end{equation}
The loss is again quadratic in the price discrepancy $d$.

\paragraph{Substituting the block-time move.}
Over a block of length $B$ the log-price move has standard deviation
\begin{equation}
  \sigma_b = \sigma\sqrt{B},
  \label{eq:sigmab}
\end{equation}
where $\sigma$ is the volatility per unit time. The price discrepancy scales with $\sigma_b$, so the event loss of \eqref{eq:levent} is of order $\ell\sigma_b^2 = \ell\sigma^2 B$. Aggregating over the arbitrage events, the arbitrage value captured on a pool per block is
\begin{equation}
  \mathrm{ARB}(B) = \frac{\ell\,\sigma_b^2}{2 + \sqrt{2}\,\gamma/\sigma_b},
  \label{eq:arb}
\end{equation}
where the denominator encodes how often the price clears the fee band~\cite{milionis2022lvr,nezlobin2025lvr}.

When the fee is small next to the per-block move ($\gamma \ll \sigma_b$), the denominator is constant and
\begin{equation}
  \mathrm{ARB}(B) \approx \tfrac{1}{2}\ell\,\sigma^2 B,
  \label{eq:arblin}
\end{equation}
linear in the block time, so $\kappa$ in \eqref{eq:sqrtt} is $\ell/2$. When the fee dominates ($\gamma \gg \sigma_b$), the value grows faster, as $\sigma_b^3 \propto B^{3/2}$. The extra half-power is from the frequency of arbitrage events. When the no arbitrage band is wide, the chance per block falls as $\sigma_b$. Therefore, the loss $\sigma_b^3$. In the case of a narrow band, the price exits almost always, which leaves the loss linear in $B$. The common arbitraged pools sit near the linear case, because the price move exceeds the fee on the volatile pools searchers chase. We do not check this pool by pool. The winning-fee exponent measured near one at high competition (Section~\ref{sec:theory}) is consistent with a linear value.

\paragraph{Value against rate.}
Normalizing \eqref{eq:arb} by the block time $B$ gives an arbitrage rate that is constant when the fee is negligible and falls as $\sqrt{B}$ in the fast-block limit. The value the winning searcher contests is the per-opportunity quantity \eqref{eq:arblin}, which scales as $B$.

\paragraph{The jump floor.}
This diffusive term vanishes as $B \to 0$. A discontinuous price move creates an arbitrage regardless of how short the window is, which contributes to the floor $J$ in \eqref{eq:sqrtt}. Bundi~\cite{bundi2026optimal} models the reference price as a jump-diffusion and separates this invariant jump channel from the diffusion channel above. We take $J$ from that decomposition.

\section{Extracted fields}
\label{app:fields}

\textbf{Data construction.} We extract the transaction data with cryo in chunks of 1,000 blocks and
map them to daily files sorted by $(\texttt{block\_number}, \texttt{tx\_index})$. The transaction
pull contains the seventeen fields listed in Table~\ref{tab:fields}, and a companion block-level pull
supplies the base fee. Integer-valued fee and value fields are in wei.

\begin{table}[h]
  \centering
  \caption{Transaction fields extracted per record, in stored order.}
  \label{tab:fields}
  \begin{tabular}{ll}
    \toprule
    Field & Description \\
    \midrule
    \texttt{block\_number}            & 2~s block height \\
    \texttt{tx\_index}                & position within the block (execution order) \\
    \texttt{tx\_hash}                 & transaction hash \\
    \texttt{block\_time}              & block timestamp (UTC) \\
    \texttt{tx\_type}                 & 0 legacy, 2 EIP-1559, 126 system/deposit \\
    \texttt{from}                     & sender address \\
    \texttt{to}                       & recipient address \\
    \texttt{value}                    & native value transferred (wei) \\
    \texttt{gas\_limit}               & gas limit \\
    \texttt{gas\_used}                & gas consumed \\
    \texttt{gas\_price}               & effective gas price paid (wei) \\
    \texttt{max\_fee\_per\_gas}       & EIP-1559 fee cap (wei) \\
    \texttt{max\_priority\_fee\_per\_gas} & declared priority-fee bid $b$ (wei) \\
    \texttt{priority\_fee\_per\_gas}  & effective priority fee $p = \texttt{gas\_price} - \texttt{base\_fee}$ (wei) \\
    \texttt{nonce}                    & sender nonce \\
    \texttt{method\_id}               & first four calldata bytes (function selector) \\
    \texttt{success}                  & execution status (true/false) \\
    \bottomrule
  \end{tabular}
\end{table}

\section{Live priority-queue pull}
\label{sec:livepull}

The slice count reflects how the builder sequences a window. It sorts the available transactions,
begins executing, and serves a newly arrived higher-fee transaction ahead of the next sorted item,
pulling live from a priority queue. This predicts the fee sequence. A live preemption at position
$i$ produces a cross-sender jump ($p_i > p_{i-1}$) whose successor resumes below the spike
($p_{i+1} < p_{i-1}$), a \emph{return}; a static per-window re-sort instead stays high
($p_{i+1} \ge p_{i-1}$), \emph{elevated}. Table~\ref{tab:livepull} classifies the post-flashblock
jumps: returns occur $6.1$ times per block and are absent pre-flashblocks, direct evidence of
mid-window preemption.
Figure~\ref{fig:trace} contrasts a representative block from each regime. This evidence is
indirect: sequencer receive-times are unavailable, so arrival dynamics are read from the fee
sequence.

\begin{figure}[t]
  \centering
  \includegraphics[width=0.80\textwidth]{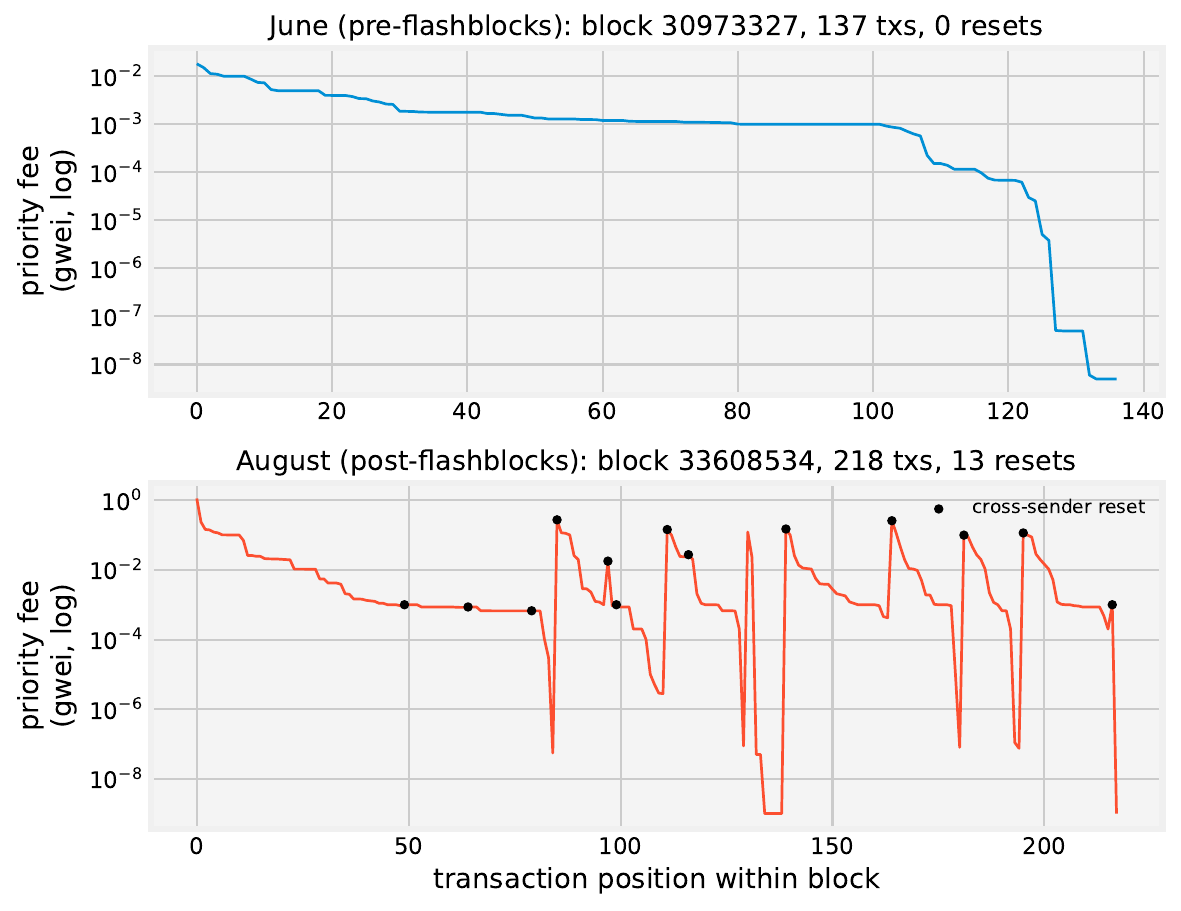}
  \caption{Effective priority fee by transaction position for a representative block per regime
  (log scale). Pre-flashblocks is monotone descending; post-flashblocks shows high-fee spikes that
  return beneath the interrupted run. Black markers are cross-sender resets.}
  \label{fig:trace}
\end{figure}

\begin{table}[t]
  \centering
  \caption{Live-pull test on post-flashblock cross-sender jumps; the pre-flashblock month has none.}
  \label{tab:livepull}
  \input{livepull_table}
\end{table}

\section{Regime stability of flashblocks}
\label{sec:stability}
In this section, we contrast the outcomes by comparing them at different times after the launch of flashblocks.

\subsection{Comparison against a mature post-flashblock month.}
\label{sec:stability_june}
August 2025 is only one month after launch, so the measured contrast might capture a builder or a searcher population that had not yet settled rather than the steady-state flashblock regime. We now extend the analysis to two further post-launch windows, October 2025 (three months after launch) and June 2026 (roughly eleven months after launch). The data is drawn and processed identically to the 2025 windows (Section~\ref{sec:data}). If the August measurement is a transition artifact, the later regimes should differ. However, if the flashblock effect is structural, the mature months should reproduce it.

Table~\ref{tab:stability} sets the three post-flashblock months against the June 2025 pre-period. The
qualitative regime is unchanged eleven months on. The single 2~s auction is absent in all post
months. The auction remains fragmented into roughly ten cross-sender sub-block contests, with the cross-sender share of upward fee jumps at 94.4\% in August, 92.8\% in October, and 93.6\% in June 2026. The fee concentration remains far below its pre-flashblock level, with the top gas-decile share $C$ at 0.27, 0.27, and 0.36 respectively, against 0.89 in June 2025.

We now discuss the two outcomes that move across the post months. This likely indicates that the builder's sequencing matured rather than an effect that decayed. First, the number of segments settles onto ten. In August 98.3\% of blocks carry eleven or more inferred segments, with an uncapped mean near 19.5; October 2025 is the intermediate point, with 66.7\% of blocks above ten and a mean of 12.0; and by June 2026 91.8\% of blocks carry exactly ten segments and only 3.3\% exceed ten, so the mean falls monotonically, 19.5 to 12.0 to 10.0. Second, and relatedly, the live priority-queue pull of Appendix~\ref{sec:livepull} disappears: the share of cross-sender jumps classified as returns falls monotonically from 33.2\% (6.1 per block) in August, through 19.8\% (2.2 per block) in October, to 0.4\% (near zero per block) in June 2026.
 The live-pull dynamic is predominantly a property of the early post-launch regime specifically, not a permanent feature of the auction.

The searcher-level fee outcomes are not comparable in levels across the post periods for the same confounding reasons discussed in Section~\ref{sec:ident}. Over eleven months the ETH price, congestion, the base fee, and the intensity of automated activity all move, and any combination of them shifts realized fees independently of the auction cadence. Descriptively, the strict-core searcher effective priority fee falls sharply from 0.45 gwei in the pre-period to 0.11 gwei in August and then holds essentially flat across the three post months (0.112, 0.110, and 0.123 gwei in August, October, and June 2026; Panel~B of Table~\ref{tab:stability}), and the searcher revert rate stays high and comparable across them (0.56, 0.49, 0.48). That the post-regime fee level is stable in gwei even as the ETH price roughly halved over the year is notable, but cross-period levels remain confounded; the within-period relative structure that the difference-in-differences design exploits is what identifies the causal fall from the pre-period.

\begin{table}[t]
  \centering
  \caption{Regime stability across the three post-flashblock months, against the June 2025
  pre-period, on the complete months. \emph{Panel A} is block-level microstructure across all senders
  (reconstruction of Section~\ref{sec:recon}; $C$ is the top gas-decile share of within-block
  priority-fee value; histogram segment counts are capped at eleven, the mean row is uncapped).
  \emph{Panel B} is the strict searcher population (screened in-period, per month, since the
  balanced-panel core of Section~\ref{sec:groups} cannot be defined for the later months) and ETH
  realized volatility; ``corr'' is the within-month daily Spearman correlation of the row outcome
  with daily realized volatility.}
  \label{tab:stability}
  \small
  \setlength{\tabcolsep}{4pt}
  \begin{tabular}{lrrrr}
    \toprule
     & \multicolumn{1}{c}{\shortstack{June 2025\\(pre)}}
     & \multicolumn{1}{c}{\shortstack{Aug 2025\\$(+1)$}}
     & \multicolumn{1}{c}{\shortstack{Oct 2025\\$(+3)$}}
     & \multicolumn{1}{c}{\shortstack{June 2026\\$(+11)$}} \\
    \midrule
    \multicolumn{5}{l}{\emph{Panel A. Block-level microstructure (all senders)}} \\
    \addlinespace
    Blocks                          & 1.29M   & 1.33M   & 1.34M   & 1.29M \\
    Mean segments/block             & 1.00    & 19.5    & 12.0    & 10.0 \\
    Single-segment share            & 100.0\% & 0.1\%   & 0.0\%   & 0.0\% \\
    Blocks with exactly 10 segments & ---     & 0.6\%   & 19.3\%  & 91.8\% \\
    Blocks with $\ge 11$ segments   & 0.0\%   & 98.3\%  & 66.7\%  & 3.3\% \\
    Cross-sender jump share         & 0.0\%   & 94.4\%  & 92.8\%  & 93.6\% \\
    Fee concentration $C$           & 0.89    & 0.27    & 0.27    & 0.36 \\
    Live-pull returns (share)       & ---     & 33.2\%  & 19.8\%  & 0.4\% \\
    Live-pull returns per block     & 0       & 6.1     & 2.2     & $\approx 0$ \\
    \addlinespace
    \multicolumn{5}{l}{\emph{Panel B. Searcher-level outcomes and ETH volatility}} \\
    \addlinespace
    Realized volatility (annualized) & 66\%    & 70\%    & 76\%    & 68\% \\
    Searcher effective fee (gwei)    & 0.451   & 0.112   & 0.110   & 0.123 \\
    Searcher declared bid (gwei)     & 0.478   & 0.107   & 0.123   & 0.123 \\
    Searcher revert rate             & 0.615   & 0.561   & 0.491   & 0.483 \\
    corr(vol, eff.\ fee), daily      & 0.81    & 0.29    & 0.28    & 0.76 \\
    corr(vol, revert rate), daily    & $-0.08$ & $-0.38$ & $-0.40$ & $-0.05$ \\
    \bottomrule
  \end{tabular}
\end{table}

Realized volatility differs across the windows, from 66\% to 76\% annualized (Panel~B of
Table~\ref{tab:stability}). The outcome nets out the base fee. What matters is the priority-fee response to volatility. The day
effects absorb the part common to both groups. June and August had comparable volatility. A volatility
gap is therefore not driving the estimate (Appendix~\ref{sec:vol}).

\clearpage
\section{Estimation details and robustness}
\label{app:robustness}

Table~\ref{tab:did-robust} compares estimates for the two control definitions in
Section~\ref{sec:groups}. The fee and concentration estimates retain their signs
under the activity-matched and retail controls.
The fee declines are larger against the activity-matched control.
The revert-rate estimates remain descriptive because their pre-trends are nonzero
(Section~\ref{sec:did}).

\begin{table}[htbp]
  \centering
  \caption{Robustness of the DiD to the control definition. Fee outcomes in gwei.}
  \label{tab:did-robust}
  \input{did_robust}
\end{table}

\textbf{Fixed-effects weighting.} The DiD column is the fixed-effects coefficient $\hat\tau$ of
\eqref{eq:twfe}, not the arithmetic difference of the two $\Delta$ columns. The two coincide only on a balanced, equally weighted panel; here the day effects and the unbalanced per-address day coverage weight each address-day cell by its within-group residual variation, so $\hat\tau$ departs from
(searcher~$\Delta$) $-$ (control~$\Delta$) in either direction: it raises the fee and bid estimates
above the raw difference ($-0.187$ vs $-0.168$; $-0.228$ vs $-0.183$) and lowers the transaction-count
estimate below it ($+55.1$ vs $+155.7$).

\textbf{Population definitions.} The searcher fee levels in this paper come from three
populations: the competitive core addresses active in both months, used for the causal estimates
($0.317$ gwei in Section~\ref{sec:did}); the in-period monthly screen (Table~\ref{tab:stability});
and the both-months aggregate (Table~\ref{tab:groups}). They are not directly comparable.
The fee-concentration index $C$ likewise appears for three populations: the balanced-panel core
($0.98$ pre-period, Section~\ref{sec:did}), the in-period searchers ($0.93$ pre-period,
Appendix~\ref{sec:hyp}), and all senders ($0.89$ pre-period, Table~\ref{tab:stability}).

\section{Opportunity-level analysis}
\label{app:oppmatch}

Grouping competitors is not trivial, because searchers execute through their own contracts: two
searchers racing the same opportunity call different addresses with different method selectors, so a
destination-and-selector key splits genuine competitors apart, and a shared router pools unrelated
transactions together. The pattern is regime-dependent: over June 2025, 90\% of core searcher fee
value flows through routers shared by two or more searchers (one contract alone accounts for 49\%),
whereas over August 2025 the majority, 56\%, flows through single-searcher contracts.

We therefore identify an opportunity by the underlying liquidity pool. Decoded \texttt{Swap} events
are pulled for the major Base decentralized exchanges (the Uniswap v2 and v3 families, the Solidly
and concentrated-liquidity Aerodrome pools, and the Uniswap v4 singleton); the emitter of a
\texttt{Swap} is the pool, so two transactions touching the same pool contest the same state
regardless of which private contract issued them. An opportunity is a single pool dislocation
delimited by its successful capture; a pool hosts several within a 2~s block. Because reverted
transactions emit no events, the pool of a losing (reverting) transaction is attributed to the pool
its contract reveals on successful transactions, which covers 96\% of core reverted fee
value in August.

\textbf{Partial payments.} Reverting losers are matched through the pool attribution above. Winners
pay on full execution; losers pay only for gas consumed before reverting.

\textbf{Secondary latency diagnostics.} Among the contested captures in August 2025
(Section~\ref{sec:oppmatch}), a higher-fee attempt reverts within two reconstructed
slices after its linked winner in $9.8$--$11.9\%$ of cases. The range spans the dominant-pool and proportional
rules for attributing a losing transaction's pool. The outbidding losers had bid
a median of $0.02$--$0.03$~gwei above the winner.

June 2025 has a single fee-sorted slice per 2~s block, so no higher-fee loser executes
after a lower-fee winner. The late-arrival rate is therefore $0.0\%$ by construction.
This is a mechanical consistency check. It does not independently validate the matching.

\textbf{Competition and winning fees.} Tables~\ref{tab:compete} and~\ref{tab:betac}
report the full competitor-count distribution and winning-fee means used in
Section~\ref{sec:theory}.

\begin{table}[htbp]
  \centering
  \caption{Competition per opportunity (pool within a 2~s block), both block-time regimes. $M$ is the
  number of distinct core searchers contesting an opportunity.}
  \label{tab:compete}
  \begin{tabular}{lrrrrr}
    \toprule
    Regime & Opportunities & Mean $M$ & Var/Mean & Uncontested & Captures/pool/block \\
    \midrule
    June (2~s)     & 10{,}886{,}949 & 1.41 & 3.75 & 81.4\% & 1.07 \\
    August (200~ms) &  9{,}591{,}245 & 1.56 & 2.50 & 78.0\% & 1.22 \\
    \bottomrule
  \end{tabular}
\end{table}

\begin{table}[htbp]
  \centering
  \caption{Mean winning fee $\beta(c)$ (gwei) and the number of opportunities, by competitor count
  $c$, in both regimes. The floor $\beta(1)$ is nearly the same across regimes; the premium $\beta(c) - \beta(1)$
  is much larger in the 2~s regime.}
  \label{tab:betac}
  \begin{tabular}{rrrrr}
    \toprule
     & \multicolumn{2}{c}{June (2~s)} & \multicolumn{2}{c}{August (200~ms)} \\
    \cmidrule(lr){2-3}\cmidrule(lr){4-5}
    $c$ & $\beta(c)$ & Opportunities & $\beta(c)$ & Opportunities \\
    \midrule
    $1$      & 0.177  & 6{,}313{,}848 & 0.164 & 5{,}376{,}386 \\
    $2$      & 0.824  &   555{,}731 & 0.366 &   643{,}861 \\
    $3$      & 1.916  &   183{,}241 & 0.455 &   243{,}283 \\
    $4$      & 3.217  &    81{,}769 & 0.495 &   127{,}281 \\
    $5$      & 4.586  &    43{,}066 & 0.638 &    74{,}055 \\
    $\ge 6$  & 11.446 &    79{,}967 & 1.027 &   132{,}027 \\
    \bottomrule
  \end{tabular}
\end{table}

\section{Transaction-level hypotheses}
\label{sec:hyp}

In addition to the address-level causal analysis, we test five transaction-level hypotheses.

\begin{description}
  \item[H1 (bidding), supported.] Preemption carries a premium. Splitting post-flashblock
        transactions by the reconstructed reset flag, the transactions that reset the queue pay
        $9.9\times$ the effective fee of interior transactions ($0.139$ versus $0.014$ gwei) and
        overbid at $2.0\times$ the interior rate. The fee reduction of Section~\ref{sec:did} thus
        coexists with a concentrated premium paid by the transactions that jump the queue.
  \item[H2 (bidding), refuted.] The within-block fee distribution does not become bimodal: the
        mean bimodality coefficient is $0.89$ in both regimes and already saturated by right skew,
        so no shape change is detectable at the block level even though a separated preemptor
        population exists.
  \item[H3 (contention), supported.] A blockspace Jevons effect holds. Reverts per block rise from
        $17.6$ to $25.1$ ($+43\%$), the revert rate from $0.078$ to $0.114$ ($+46\%$), and burned
        revert-gas per block from $2.5$M to $3.4$M ($+36\%$), while reverts per inferred contest
        fall from $17.6$ to $1.3$ ($-93\%$). Cheaper, faster retries multiply attempts per block
        even as each sub-contest is far less contested.
  \item[H4 (timing), refuted.] Inter-reset spacing does not synchronize to the 200~ms cadence. The
        coefficient of variation of inter-reset gaps is $1.17$, near the memoryless value of one,
        and reset positions are near-uniform with no boundary comb. This is measured on a
        gas-position proxy for arrival time, as sub-block timestamps are unavailable.
  \item[H5 (fee concentration), refuted as composition.] The concentration collapse
        is not a compositional artifact: within searchers alone $C$ falls from $0.93$ to $0.26$,
        matching the pooled fall from $0.89$ to $0.27$. This matches the causal
        fee-concentration effect estimated in Section~\ref{sec:did} ($0.98$ to $0.28$ there, measured a
        different way).
\end{description}

Figure~\ref{fig:did-revert} shows the revert-rate event study behind the non-identification noted in
Section~\ref{sec:did}.

\begin{figure}[t]
  \centering
  \includegraphics[width=0.86\textwidth]{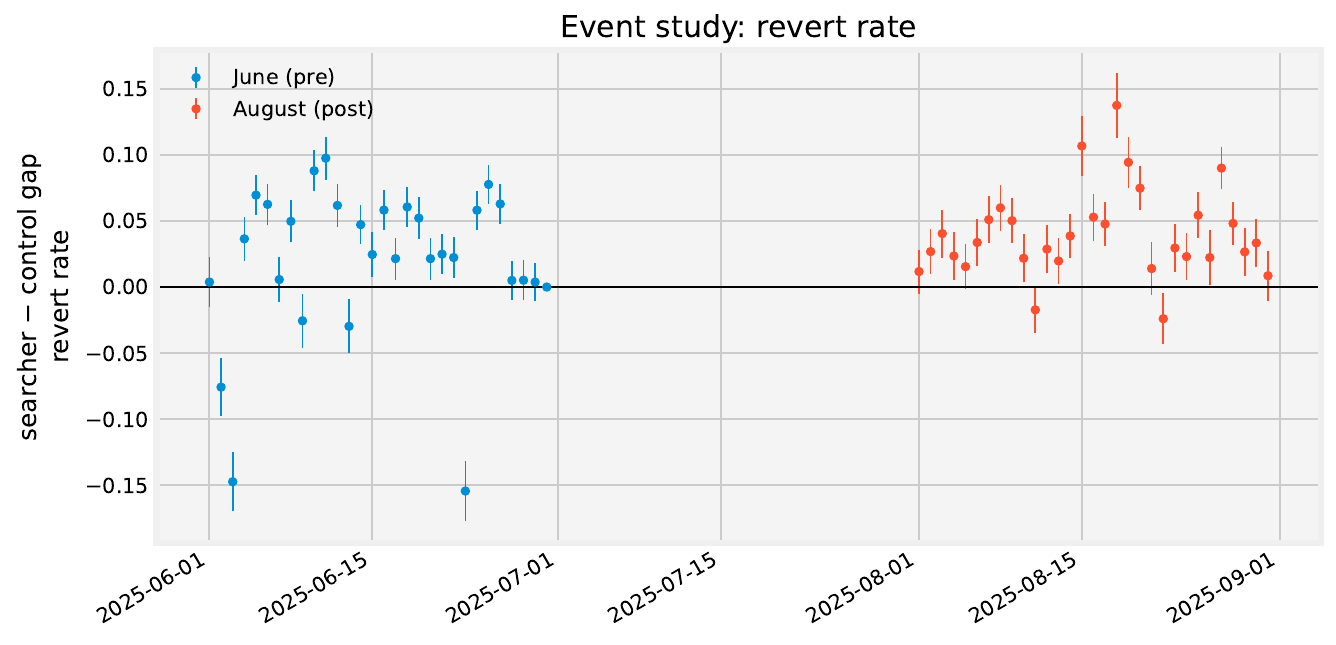}
  \caption{Event study for the revert rate. The pre-period gap is not stable, the visual
  counterpart of the failed pre-trend test that withholds a causal reading.}
  \label{fig:did-revert}
\end{figure}

\section{Market volatility as a confounder}
\label{sec:vol}

The most volatility-sensitive outcome is the fee level, which bears directly on the parallel-trends
assumption of Section~\ref{sec:ident}, so the four windows are checked against ETH volatility.
Minute-level realized volatility, computed from Binance one-minute klines, differs materially across
the windows (Panel~B of Table~\ref{tab:stability}): annualized, it is 66\% in June 2025, 70\% in
August 2025, 76\% in October 2025, and 68\% in June 2026. October is the most volatile month, driven
by a single extreme day, the liquidation cascade of 2025-10-10, at 217\% annualized against the 76\%
monthly mean.

At daily resolution the searcher effective fee tracks realized volatility within every
month, whereas contention tracks it only weakly. The effective-fee correlation ranges from 0.28 to
0.81, higher in the two June windows than in August and October, while the revert-rate correlation is small and, for this strict core, mildly
negative (pooled $-0.18$): the core group's revert rate is already saturated and does not rise with
volatility. On the October cascade the searcher effective fee stood 4.7 standard deviations above its
monthly mean while the revert rate and transaction count barely moved, so market stress raises the
auction's clearing price rather than the volume of failed racing.

The outcome nets out the base fee, so a base-fee move cannot bias it. The concern is the priority-fee
response to volatility. The day effects absorb the part common to both groups. Parallel trends fails
only if volatility moves the two groups' priority fees by different amounts. The near-zero control fee
makes that hard to check. Two facts limit the concern. June and August had comparable volatility, $66\%$
against $70\%$ annualized. A volatility gap between them is therefore not driving the estimate. The event
study steps down at the July launch. It does not track volatility day to day (Figure~\ref{fig:did-eff}).

\section{Extrapolating bids below the 200~ms window}
\label{sec:structural}
\label{app:extrap}
This section projects searcher bidding to flashblock windows shorter than the observed 200~ms. The
difference-in-differences of Section~\ref{sec:did} identifies the searcher bid at the two cadences
realized on-chain, the single 2~s auction and the ten 200~ms flashblocks. It does not identify the bid
below 200~ms. A natural counterfactual is a still-shorter preconfirmation interval, so we ask
what the structure of the sub-auction implies at 125, 100, and 50~ms. The realized within-August
segment count cannot serve as the regressor. It varies with contention rather than window length, and
the induced bid-versus-count relation is non-monotone and wrong-signed. We therefore calibrate a
micro-auction model on the two observed points.

Write the 2~s block as a window of length $T = 2000$~ms split into $n$ equal sub-auctions of length
$B = T/n$, each including the highest-bid transaction,
\begin{equation}
  n \;=\; T / B,
  \label{eq:subauc}
\end{equation}
where $B$ is the sub-auction length and $T = 2000$~ms the block window. Then $n = 1$ at the 2~s
auction, $n = 10$ at 200~ms, and the targets $B = 125, 100, 50$~ms correspond to $n = 16, 20, 40$.
Competing searchers for a given opportunity arrive over the block as a Poisson process with expected
count $\Lambda$ per 2~s window. Thinned across the $n$ sub-windows, the expected number contesting any
one slice is
\begin{equation}
  c(n) \;=\; \Lambda / n ,
  \label{eq:thinning}
\end{equation}
where $\Lambda$ is the expected number of competitors per opportunity per 2~s window. A searcher's
winning bid carries a premium over the competitive floor that increases in the number of rivals it
must outbid within its slice. The premium falls as $n$ grows and vanishes in the continuous-auction
limit $n \to \infty$.

Let $\pi(n)$ be the searcher declared bid net of the activity-matched control at cadence $n$, the
control differencing out shocks common to both groups as in \eqref{eq:twfe}. The two on-chain cadences
give the searcher-minus-control levels underlying the Section~\ref{sec:did} estimates,
\begin{equation}
  \pi(1) \;=\; 0.374~\text{gwei}, \qquad \pi(10) \;=\; 0.192~\text{gwei},
  \label{eq:anchor}
\end{equation}
where $\pi(1)$ and $\pi(10)$ are the control-adjusted declared bids at the 2~s and 200~ms cadences,
with effective-fee analogs $0.316$ and $0.148$~gwei.\footnote{These are group-level means. Their
double difference $\pi(1) - \pi(10) = 0.182$~gwei is close to but not identical to the regression
estimate $\hat\tau = -0.228$ of Table~\ref{tab:did}. The gap is not the group-size imbalance
($3{,}032$ versus $8{,}053$), which cancels in either group mean. It is the two-way fixed-effects
weighting, which averages each address-day cell by its within-group residual variation on an
unbalanced panel, with addresses active on differing numbers of days and the day effects reweighting
the pre-to-post contrast. The model is calibrated on the level means.} We assume the premium vanishes
in the continuous limit, $\pi(\infty) = 0$. We fit these observations with three shapes,
\begin{equation}
  \pi(n) \;\in\; \big\{\, A\,n^{-\gamma}, \;\; A\big(1 - e^{-\Lambda/n}\big), \;\; A\,\log\!\big(1 + \Lambda/n\big) \,\big\},
  \label{eq:shapes}
\end{equation}
where $A$ is a scale, $\gamma$ a decay exponent, and the three forms are the power (rank-rent),
Poisson-hazard, and logarithmic shapes. The power shape carries no $\Lambda$ and is fixed exactly by
the two points in \eqref{eq:anchor}. The hazard and logarithmic shapes are fit for $(A, \Lambda)$ on
the same two points, and their fitted $\Lambda$ is checked against the competition intensity measured
directly. Matching core searcher transactions to the pools of Section~\ref{sec:oppmatch}, the mean
number of distinct core searchers contesting a pool within a 2~s block is $\Lambda = 1.56$, and $3.56$
among opportunities contested by two or more.

Table~\ref{tab:structural} reports the projected premium. At the 200~ms, all three shapes agree by
construction. Below it they diverge. The declared-bid premium at 50~ms ranges from $0.128$~gwei under
the power shape to $0.061$~gwei under the hazard shape, a $33$ to $68\%$ reduction relative to the
200~ms level, and the effective fee tracks it.

Reading the projected changes requires two checks. First, the Poisson-thinning shapes require $\Lambda$ between
$6$ and $92$ to fit the two points, far above the measured $1.56$. At the measured intensity the
per-slice competition at $n = 10$ is only $\Lambda/10 \approx 0.16$. A pure thinning model then
predicts a premium near zero at 200~ms, whereas the observed premium is still half its 2~s value. The
surviving 200~ms premium is therefore a floor set by a fixed inclusion cost rather than by any
competition. Only the power shape, which carries no $\Lambda$, is consistent with the measured
intensity. Second, we split the core searchers into flashblock-sensitive and insensitive subgroups by the
proportional June-to-August bid change. The sensitive-minus-insensitive premium falls from $0.49$ to
$0.02$~gwei between $n = 1$ and $n = 10$. This is steeper than any thinning shape permits, so the
responsive searchers already sit at their competitive floor at 200~ms.

The projection rests on three assumptions. We assume the competition intensity $\Lambda$ and the value
distribution are invariant to the cadence, and that the sub-auctions are one-shot and independent
. None of these is testable
within two calibration points, so the sub-200~ms figures are model extrapolation
rather than identified estimates. Both diagnostics point the same way. The first-order bid compression
from flashblocks is floored at 200~ms, and shortening the window further yields diminishing returns.
\begin{table}[t]
  \centering
  \caption{Model-extrapolated searcher bid premium $\pi(n)$ (searcher minus control, gwei) below the
  200~ms window, by shape, with $n = 2000/B$ sub-auctions. The $200$~ms column is the calibration
  anchor \eqref{eq:anchor} and is common to all shapes; the sub-200~ms columns are extrapolation. The
  Fit column reports $\gamma$ (power) or the $\Lambda$ needed to fit the two anchors (hazard,
  logarithmic), against the measured $\Lambda = 1.56$.}
  \label{tab:structural}
  \begin{tabular}{lrrrrr}
    \toprule
     & \multicolumn{4}{c}{Window $B$ (ms) / sub-auctions $n$} & \\
    \cmidrule(lr){2-5}
    Shape & 200 / 10 & 125 / 16 & 100 / 20 & 50 / 40 & Fit \\
    \midrule
    \multicolumn{6}{l}{\emph{Panel A. Declared bid $b$}} \\
    \addlinespace
    Power $A\,n^{-\gamma}$          & 0.192 & 0.167 & 0.157 & 0.128 & $\gamma = 0.29$ \\
    Hazard $A(1-e^{-\Lambda/n})$    & 0.192 & 0.135 & 0.113 & 0.061 & $\Lambda = 7.2$ \\
    Log $A\log(1+\Lambda/n)$        & 0.192 & 0.158 & 0.142 & 0.098 & $\Lambda = 92$ \\
    \addlinespace
    \multicolumn{6}{l}{\emph{Panel B. Effective priority fee $p$}} \\
    \addlinespace
    Power $A\,n^{-\gamma}$          & 0.148 & 0.127 & 0.118 & 0.094 & $\gamma = 0.33$ \\
    Hazard $A(1-e^{-\Lambda/n})$    & 0.148 & 0.103 & 0.086 & 0.046 & $\Lambda = 6.3$ \\
    Log $A\log(1+\Lambda/n)$        & 0.148 & 0.118 & 0.105 & 0.069 & $\Lambda = 57$ \\
    \bottomrule
  \end{tabular}
\end{table}

\section{Extrapolation from measured competition}
\label{sec:routeb}

This section re-derives the sub-200~ms projection from primitives measured at the 200~ms cadence, so
the extrapolated shape is built from measured primitives rather than an assumed functional form. The data reject the Poisson assumption. Among captured
opportunities the competitor count is overdispersed, with a variance-to-mean ratio of $1.29$ against
the Poisson value of one. The coarser pool-block count is more dispersed still, at $2.50$.

\textbf{Competition response.} Transactions are matched to pools as in Section~\ref{sec:oppmatch}. For
each captured opportunity, a pool within a 2~s block, let $c$ be the number of distinct core searchers
contesting it and take the winning effective fee. The response
\begin{equation}
  \beta(c) \;=\; \mathbb{E}\big[\,p_{\text{win}} \ \big|\ c \text{ competitors, opportunity value}\,\big]
  \label{eq:beta}
\end{equation}
where $p_{\text{win}}$ is the winning transaction's effective fee, is measured at the 200~ms cadence
(Table~\ref{tab:beta}, Panel~A). The raw estimate
$\operatorname{mean}_{o:\, c_o = c}\, p_{\text{win},o}$ is the mean winning fee across all opportunities
contested by $c$ searchers, at any value. Higher-$c$ opportunities skew toward more valuable pools, so
the raw estimate conflates competition with value. The value-controlled estimate conditions on value,
proxied by the winning transaction's gas used $V_o$ binned into quartiles $q(V_o) \in \{1, \dots, 4\}$.
Within each quartile we take the mean winning fee at count $c$, then average the four quartile means
equally,
\begin{equation}
  \beta(c) \;=\; \frac{1}{4}\sum_{q=1}^{4}\
    \operatorname*{mean}_{o:\, c_o = c,\ q(V_o) = q}\, p_{\text{win},o},
  \label{eq:betavc}
\end{equation}
where $q(V_o)$ is the gas-used quartile of opportunity $o$. The two estimates nearly coincide, rising
from about $0.16$ to between $0.95$ and $1.03$~gwei (Table~\ref{tab:beta}). The slope therefore
reflects competition. Value composition does not drive it. Gas used is a coarse value proxy, and a
swap-amount proxy from the pool-level data is left for future work.

\textbf{Competitor law and thinning.} Among captured opportunities the competitor count $M$ is
overdispersed, with mean $\Lambda_0 = 1.43$ and negative-binomial shape $r = 4.9$ (below the $1.56$
all-opportunity mean of Appendix~\ref{app:extrap}, since captured opportunities drop the more contested
pool-blocks where no core searcher won),
\begin{equation}
  \mathbb{E}[M] = \Lambda_0, \qquad
  M \sim \text{NB}(\Lambda_0, r),
  \label{eq:nb}
\end{equation}
where $\operatorname{Var}(M) = \Lambda_0 + \Lambda_0^2/r$, and $r \to \infty$ recovers the Poisson
case. A Gamma-Poisson mixture thins, so at cadence $n$ each of the $n$ windows carries
$K_n \sim \text{NB}(\Lambda_0/n,\, r)$, the same shape with the mean split across windows. The expected
winning fee at cadence $n$, conditional on the window being contested, is
\begin{equation}
  \hat p(n) \;=\; \mathbb{E}\big[\,\beta(K_n)\ \big|\ K_n \ge 1\,\big]
             \;=\; \frac{\sum_{c \ge 1} \Pr(K_n = c)\,\beta(c)}{\Pr(K_n \ge 1)},
  \label{eq:routeb}
\end{equation}
where $K_n$ is the per-window competitor count at cadence $n$ and the conditioning restricts to
windows holding at least one searcher.

\textbf{Result.} Table~\ref{tab:routeb} evaluates \eqref{eq:routeb}. Two features stand out. The first
is the floor. With no rival, $\beta(1) = 0.164$~gwei is a value-capture bid, so
$\hat p(\infty) = \beta(1)$ is nonzero. This corrects the vanishing-premium assumption of
Section~\ref{sec:structural}. The second is that the two observed cadences are reproduced as
consistency checks. Built only from 200~ms primitives, the model returns $0.316$~gwei at
$n = 1$, matching the 2~s regime. It returns $0.181$~gwei at $n = 10$, close to the 200~ms level. Below
200~ms the predicted fee is nearly flat, from $0.181$~gwei at 200~ms to $0.169$~gwei at 50~ms, and from
$0.188$ to $0.172$~gwei under the value-controlled response. The decline is under $10\%$ and converges
on the $0.164$ floor, and the raw and value-controlled paths agree within $0.01$~gwei. The band is an
order of magnitude tighter than in Section~\ref{sec:structural} because the shape is measured rather
than fit.

At 200~ms the expected number of rival searchers in a single window is
only $\Lambda_0/10 = 0.14$, so most windows are uncontested and the winning bid already sits at the
value-capture floor $\beta(1)$. Shortening the window removes competition that is largely gone. The
assumption needed is competition-count invariance, that a searcher facing $c$ rivals bids the same
whether its window is 200 or 50~ms. This is weaker than assuming invariance to the window length
itself. Transporting $\beta(c)$ across window lengths remains an assumption. But Section~\ref{sec:theory}
shows the same searchers bid differently at 2~s and 200~ms (Table~\ref{tab:betac}), so this does not
hold across window lengths, and the projection is illustrative rather than exact.

\begin{table}[t]
  \centering
  \caption{Measured competition response $\beta(c)$, the winning effective fee (gwei) on an
  opportunity contested by $c$ core searchers at the 200~ms cadence. Value-controlled row averages
  $\beta(c)$ within gas-used quartiles.}
  \label{tab:beta}
  \begin{tabular}{lrrrrrr}
    \toprule
    Competitors $c$ & 1 & 2 & 3 & 4 & 5 & $\ge 6$ \\
    \midrule
    $\beta(c)$ raw              & 0.164 & 0.366 & 0.455 & 0.495 & 0.638 & 1.027 \\
    $\beta(c)$ value-controlled & 0.167 & 0.408 & 0.526 & 0.520 & 0.657 & 0.945 \\
    \bottomrule
  \end{tabular}
\end{table}

\begin{table}[t]
  \centering
  \caption{Predicted winning effective fee $\hat p(n)$ (gwei) by cadence from \eqref{eq:routeb}, with
  $n = 2000/B$ windows and expected rival searchers per window $\Lambda_0/n$. The $n = 1$ and $n = 10$
  rows are consistency checks against the 2~s and 200~ms regimes; $n = 16, 20, 40$ are the
  extrapolation. $\hat p(\infty) = \beta(1) = 0.164$.}
  \label{tab:routeb}
  \begin{tabular}{rrrrr}
    \toprule
    $B$ (ms) & windows $n$ & rivals/window $\Lambda_0/n$ & $\hat p(n)$ raw & $\hat p(n)$ value-ctrl \\
    \midrule
    2000 &  1 & 1.425 & 0.316 & 0.342 \\
     200 & 10 & 0.142 & 0.181 & 0.188 \\
     125 & 16 & 0.089 & 0.175 & 0.180 \\
     100 & 20 & 0.071 & 0.173 & 0.178 \\
      50 & 40 & 0.036 & 0.169 & 0.172 \\
    \bottomrule
  \end{tabular}
\end{table}

%% file: did_results.tex
\begin{tabular}{lrrrrr}
\toprule
outcome & searcher $\Delta$ & control $\Delta$ & DiD & SE & pre-trend $p$ \\
\midrule
effective priority fee (gwei) & -0.1680 & +0.0001 & -0.1869$^{*}$ & 0.0380 & 0.785 \\
declared bid (gwei) & -0.1818 & +0.0008 & -0.2276$^{*}$ & 0.0445 & 0.976 \\
overbid (gwei) & -0.0004 & +0.0007 & -0.0004$^{*}$ & 0.0001 & 0.000 \\
revert rate & +0.0021 & -0.0031 & +0.0208$^{*}$ & 0.0047 & 0.000 \\
tx per address-day & +137.6210 & -18.0975 & +55.1013 & 25.1451 & 0.000 \\
fee concentration $C$ & -0.702 & +0.018 & -0.719$^{*}$ & 0.004 & 0.450 \\
\bottomrule
\end{tabular}

%% file: livepull_table.tex
\begin{tabular}{lr}
\toprule
 & August (post) \\
\midrule
Cross-sender jumps per block & 18.5 \\
Return: single preemption (resumes below) & 33.2\% (6.1/block) \\
Elevated: burst or boundary (stays high) & 66.4\% (12.3/block) \\
\bottomrule
\end{tabular}

%% file: did_robust.tex
\begin{tabular}{lrr}
\toprule
outcome & DiD (activity-matched) & DiD (retail) \\
\midrule
effective priority fee (gwei) & -0.1869 & -0.1786 \\
declared bid (gwei) & -0.2276 & -0.2175 \\
revert rate & +0.0208 & +0.0324 \\
fee concentration $C$ & -0.719 & -0.300 \\
\bottomrule
\end{tabular}